\newif\ifTWO \TWOfalse
\ifTWO
\documentclass{IEEEtran}
\else
\documentclass[11pt,draftcls,onecolumn]{IEEEtran}
\fi

\usepackage{amsbsy}
\usepackage{amsmath}
\usepackage{amssymb}
\usepackage{bm}
\usepackage{cite}
\usepackage{graphicx}
\usepackage{mathbbol}
\usepackage{theorem}
\usepackage{tikz}
\newcommand{\insertfig}[2]{
\begin{figure}\centering
\includegraphics[width=.95\columnwidth]{#1.pdf}
\caption{#2}\label{#1.fig}\end{figure}}

\newtheorem{theorem}{Theorem}[section]

\theorembodyfont{\rmfamily}

\newtheorem{remark}{Remark}[section]
\def\ben{\begin{enumerate}}
\def\beq{\begin{equation}}
\def\beqa{\begin{eqnarray}}
\def\bit{\begin{itemize}}
\def\een{\end{enumerate}}
\def\eeq{\end{equation}}
\def\eeqa{\end{eqnarray}}
\def\eit{\end{itemize}}

\def\non{\nonumber\\}

\DeclareMathAlphabet{\mathsfbf}{OT1}{cmss}{sbc}{n}

\def\Am{{\bm{A}}}

\def\av{\bm{a}}

\def\balpha{\bar\alpha}

\def\diag{{\rm{diag}}}

\def\e{\mathrm{e}}

\def\E{\mathbb{E}}

\def\H{\mathsfbf{H}}

\def\Hc{\mathcal{H}}

\def\hv{\bm{h}}

\def\I{\mathcal{I}}

\def\muv{\bm{\mu}}

\def\N{\mathcal{N}}

\def\Peq{\mathsf{P}_\mathsf{eq}}
\def\Pfa{\mathsf{P}_\mathsf{fa}}
\def\Pmd{\mathsf{P}_\mathsf{md}}

\def\Pm{\bm{P}}

\def\R{\mathbb{R}}

\def\Sgm{\bm{\Sigma}}

\def\SNR{\mathsf{SNR}}
\def\sv{\bm{s}}
\def\T{\mathsfbf{T}}

\def\xih{\hat{\xi}}

\def\Ym{{\bm{Y}}}

\def\yv{{\bm{y}}}

\def\Zm{\bm{Z}}

\def\zv{\bm{z}}

\def\0{\bm{0}}
\def\1{\bm{1}}
\def\<{\left\langle}
\def\>{\right\rangle}

\begin{document}
\title{A Lower Bound to the Receiver Operating Characteristic of a Cognitive 
Radio Network\\
{\Large
(submitted to the {\em IEEE Transactions on Information Theory}, July 2010)}}
\author{Giorgio Taricco}
\maketitle
\def\thefootnote{$*$}
\footnotetext{
Giorgio Taricco is currently with Politecnico di Torino (DELEN), corso Duca 
degli Abruzzi 24, 10129, Torino, Italy (e-mail:~taricco@polito.it).
}
\def\thefootnote{\arabic{footnote}}

\begin{abstract}
Cooperative cognitive radio networks are investigated by using an 
information-theoretic approach.
This approach consists of interpreting the decision process carried out at the 
fusion center as a binary (asymmetric) channel, whose input is the presence of 
a primary signal and output is the fusion center decision itself.
The error probabilities of this channel are the false-alarm and 
missed-detection probabilities.
After calculating the mutual information between the binary random variable 
representing the primary signal presence and the set of sensor (or secondary 
user) output samples, we apply the data-processing inequality to derive a lower 
bound to the receiver operating characteristic.
This basic idea is developed through the paper in order to consider the cases 
of full channel and signal knowledge and of knowledge in probability 
distribution.
The advantage of this approach is that the ROC lower bound derived is 
independent of the particular type of spectrum detection algorithm and fusion 
rule considered.
Then, it can be used as a benchmark for existing practical systems.
\end{abstract}


\section{Introduction}

Cognitive Radio (CR) technologies have gained considerable interest in the last 
few years because of two factors:
$i)$ the increasing demand for wireless spectrum from a large number of 
applications; and
$ii)$ the fact that many portions of licensed spectrum are neglected or 
underutilized by the regular 
licensees~\cite{fcc02,mit00,hay05,blum97,sri07,stev09}.

The concept of CR depends considerably on the application 
context~\cite{yucek09}.
Nevertheless, an official definition has been given by the Global Standards 
Collaboration (GSC) group within the ITU~\cite{fette09}: ``A radio or system 
that senses its operational electromagnetic environment and can dynamically and 
autonomously adjust its radio operating parameters to modify system operation, 
such as maximize throughput, mitigate interference, facilitate 
interoperability, access secondary markets.''
According to this definition, a CR device should be able to {\em autonomously} 
exploit unused portions spectrum to increase its own signalling rate without 
limiting the use of the radio spectrum from licensed users.
Thus, the most important feature of a CR device is the ability to detect the 
availability of {\em spectrum holes}~\cite{yucek09}, which can be accomplished 
by suitable {\em spectrum sensing} techniques.
The key role of spectrum sensing has been recognized in the technical 
literature as the enabling technique for CR systems.
Different strategies have been envisaged to effectively implement this feature 
and a comprehensive taxonomy can be found in~\cite{yucek09}.

A simple statement of the CR detection problem can be given as follows.
In a CR network there are two classes of users:
$i)$ {\em primary users}, i.e., those users who have license rights of some 
other form of priority with respect to the radio channel access;
$ii)$  {\em secondary users}, i.e., those users who have no licence rights or 
have more limited priority to the channel access than the primary users.
Secondary users are those who need CR capabilities, such as spectrum sensing, 
in order to avoid causing interference to primary users.
Thus, secondary users have to estimate the radio channel condition before 
attempting a transmission, i.e., they need to assess whether the channel is 
idle or busy (hypotheses $\Hc_0$ and $\Hc_1$, respectively).
This estimation is usually affected by error and characterized by two error 
probabilities:
\bit
\item
The {\em false-alarm} probability $\Pfa$, corresponding to the detection of 
hypothesis $\Hc_1$ when $\Hc_0$ is true.
\item
The {\em missed-detection} probability $\Pmd$, corresponding to the detection 
of hypothesis $\Hc_0$ when $\Hc_1$ is true.
\eit
Ideally, secondary users should operate toward reaching the goal of having 
$\Pfa=\Pmd=0$.
However, radio channel impairments prevent to attain this operating level, and 
a suitable tradeoff has to be sought.
Typically, secondary users are allowed a maximum level of interference to the 
primary users, which translates into a maximum probability of missed 
detection.
Then, the CR users can maximize their throughput by maximizing the false-alarm 
probability $\Pfa$ under the constraint of a given $\Pmd$.
Typical values of these probabilities have been set to $\Pmd=\Pfa=0.1$ in the 
contest of the developing standard IEEE 802.22~\cite{stev09}.
A complete picture of the performance of a CR system is provided by the 
receiver operating characteristic (ROC) plot.
The ROC is a plot of the missed-detection probability $\Pmd$ versus the 
false-alarm probability $\Pfa$.
Its derivation depends on the radio channel parameters (fading, noise power) 
and on the type of decision process implemented to detect the presence of a 
primary signal.

It has been widely recognized in the literature (see, {\em 
e.g.,}~\cite{yucek09} and references therein) that {\em user cooperation} 
enhances the performance of a CR system, both in terms of ROC, and of avoiding 
the {\em hidden primary user problem}.
This problem is considered one of the major challenges to the implementation of 
a CR system, and is similar to the hidden node problem experienced in Carrier 
Sense Multiple Accessing (CSMA)~\cite{yucek09}.
The hidden primary user problem derives from the shadowing of secondary users, 
occurring while sensing the primary signal transmission.
More precisely, a secondary user can be in the range of a primary user {\em 
receiver} but out of the range of another primary user {\em transmitter}.
Then, the secondary user senses the channel idle, because it cannot capture the 
primary user signal, and then starts its transmission.
However, since it is in the range of the other primary user receiver, it 
eventually interferes with the reception of the primary signal.
Having multiple secondary users sensing the channel reduces the chances of 
falling into this situation.

Fig.~\ref{block.fig} illustrates the block diagram of a CR system based on user 
cooperation.
We can see that the primary signal is present if $\xi=1$.
This signal is received by a set of $K$ secondary users (or sensors) which 
sample it during a certain observation window.
Secondary users can exploit individually this information in order to make a 
decision on the spectrum availability.
Otherwise, they can share it by sending a suitable signal through a control 
channel to a central processing unit (i.e., implementing user cooperation).
This unit provides for the {\em fusion} of the user information and is then 
called {\em fusion center} (FC)~\cite{yucek09}.

The goal of this paper is to analyze, by using information-theoretic results, 
the behavior of a cooperative CR system.
Our approach is based on the observation that appending the decision process 
implemented at the FC to the primary signal transmission channel yields an 
equivalent binary channel with input $\xi$, the random variable indicating the 
signal presence, and output $\xih$, the FC decision.
According to this interpretation, the false-alarm and missed-detection 
probabilities correspond to the two error probabilities of this binary channel 
(conditioned to $\xi=0$ and $\xi=1$).
In general, this channel turns out to be an asymmetric binary channel because 
the error probabilities are different.

Then, by using the data-processing inequality, we can calculate an upper bound 
to the channel capacity, which translates into a lower bound to the ROC.
This basic idea is developed through the paper in order to derive the lower 
bound of a cooperative CR system ROC, which is independent of the spectrum 
detection and fusion strategy used.
This lower bound can be applied to assess the validity of specific combinations 
of spectrum sensing and fusion strategy.

The remainder of the paper is organized as follows.
The system model is illustrated in Section \ref{system.sec}, where the key 
concept of applying the data-processing inequality to the cooperative CR system 
is introduced and analyzed in detail.
Section \ref{mi.sec} deals with the derivation of the mutual information of the 
cooperative CR system without the FC channel and detection.
This section considers the case of known channel gains and signal, as a 
baseline, and the case of known channel gain and signal distribution, as a 
further development.
Relevant asymptotic cases are also studied, in order to mitigate the numerical 
difficulties in the derivation of the results.
Section \ref{results.sec} illustrates the analytic results through numerical 
examples including.
Lower bounds to the ROC are reported in this section along with a comparison of 
these results with an energy detection estimator.
Finally, our conclusions are collected in Section \ref{concl.sec}.

\section{System model and optimum ROC}\label{system.sec}

\begin{figure}
\begin{center}\begin{tikzpicture}[>=stealth]
\path
(0,2) node(xi){$\xi\in\{0,1\}$}
(2,2) node(x)[circle,draw,fill=blue!30,inner sep=0mm]{$\times$}
(2,3.5) node(ptx)[rectangle,draw,fill=green!30,inner sep=1mm]
{\parbox{1.5cm}{\centering\sffamily\small Primary \\signal}}
(5,4) node(S1)[rectangle,draw,fill=blue!30,inner sep=1mm]{$S_1$}
(5,3) node(S2)[rectangle,draw,fill=blue!30,inner sep=1mm]{$S_2$}
(5,2) node{$\vdots$}
(5,1) node(SK)[rectangle,draw,fill=blue!30,inner sep=1mm]{$S_K$}
;
\draw[thick,->] (ptx) -- (x);
\draw[thick,->] (xi) -- (x);
\draw[thick,dotted,->] (x)--++(1.5,0)--++(0,2)--(S1);
\draw[thick,dotted,->] (x)--++(1.5,0)--++(0,1)--(S2);
\draw[thick,dotted,->] (x)--++(1.5,0)--++(0,-1)--(SK);
\draw[thick,->] (S1)--++(1,0)node[right]{$y_1(n)$};
\draw[thick,->] (S2)--++(1,0)node[right]{$y_2(n)$};
\draw[thick,->] (SK)--++(1,0)node[right]{$y_K(n)$};
\end{tikzpicture}\end{center}
\caption{
Block diagram of a cognitive radio system with input $\xi$, denoting the 
primary signal presence or absence, and output given by the set of sensor 
outputs $y_k(n)$ for $k=1,\dots,K$ and $n=1,\dots,N$.
Dotted lines represent the fading channels connecting the primary transmitter 
to the sensors (secondary users).}
\label{block.fig}
\end{figure}
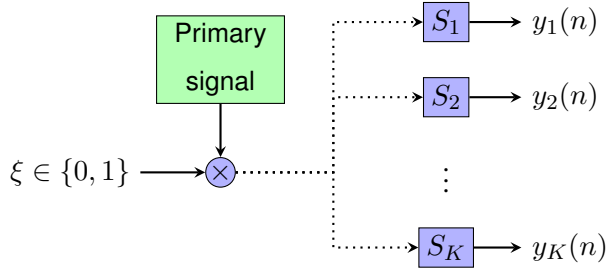

We consider a CR system (illustrated in Fig.\ \ref{block.fig}) equipped with 
$K$ sensors sensing the wireless spectrum over $N$ sampling times in order to 
provide information about the channel availability to secondary users.
Intentionally, the diagram does not show the terminal part consisting in the 
collection of the sensor measurements, their compacting, their transmission to 
the FC through a control channel, and the FC processing block providing the 
output decision about the signal presence.

We assume a block fading channel where the $n$th sampled signal received by 
sensor $k$ is given by
\beq\label{xk.eq}
  y_k(n) = \left\{\begin{array}{ll}
            z_k(n) & \xi=0 \\
    h_ks(n)+z_k(n) & \xi=1
  \end{array}\right.
\eeq
for $k=1,\dots,K$ and $n=1,\dots,N$.
Here, $z_k(n)\sim\N_c(0,\sigma_k^2)$%
\footnote{
Notation $\zv\sim\N_c(\muv,\Sgm)$ denotes a circularly symmetric complex 
Gaussian distributed vector with mean $\muv$, covariance matrix 
$\Sgm=\E[\zv\zv^H]-\muv\muv^H$, and pdf 
$\det(\pi\Sgm)^{-1}\exp[-(\zv-\muv)^H\Sgm^{-1}(\zv-\muv)]$.
}
are the iid received noise samples, $h_k$ are the block fading gain 
coefficients, $s(n)$ are the primary user's symbols, and $\xi$ denotes the the 
random variable indicating that the primary signal is present ($\xi=1$) or 
absent ($\xi=0$).
The variances $\sigma_k^2$ are known parameters.
We can interpret $\xi$ as the imponderable primary user decision to convey 
information through the channel at the time the CR system is trying to check 
the existence of a {\em spectrum hole}.

In the following we assume that the random variable $\xi$ is not necessarily 
equiprobable but rather we have $P(\xi=0)=\alpha$.
Then, $\alpha$ represents the {\em a priori} probability of primary signal 
absence.

As already mentioned, in this framework we do not consider the remaining part 
of the communication system beyond the block diagram of Fig.\ \ref{block.fig}.
This part consists of a distributed algorithm at the $K$ sensors and at the 
fusion center (FC) aimed at condensing the available channel sensing 
information (at the sensors), sending it to the FC, and jointly processing in 
order to make a reliable decision on the presence of a primary transmitted 
signal.

On the contrary, we regard the block diagram in Fig.\ \ref{block.fig} as a 
binary input-continuous output vector channel, which we study in order to 
derive the mutual information
\beq\label{ycap.eq}
  \I \triangleq I\Big(\xi;\{y_k(n)\}_{k=1,n=1}^{K,N}\Big).
\eeq
Using the {\em data processing inequality} (DPI)~\cite{cover}, we can see that 
the mutual information $\I$ upper bounds the mutual information of the channel 
corresponding to the completion of the transmission chain to the FC by any 
conceivable distributed algorithm.

Completing the transmission chain up to the FC's output yields a binary-input 
binary-output channel.
Denoting the FC's output by $\xih$, we have from the DPI:
\beq\label{dpi.eq}
  I(\xi;\xih) \le \I.
\eeq
Now, by the definition of the false-alarm and missed-detection probabilities 
(denoted by $\Pfa$ and by $\Pmd$, respectively), we have:
\[
\left\{\begin{array}{lll}
  \Pfa &=& P(\xih=1\mid\xi=0) \\
  \Pmd &=& P(\xih=0\mid\xi=1)
\end{array}\right..
\]
In general, we have a binary asymmetric channel whose transition probability 
matrix can be written as
\[
  \Pm = \begin{pmatrix}
  1-\Pfa & \Pfa \\ \Pmd & 1-\Pmd
  \end{pmatrix}.
\]
The mutual information, assuming $P(\xi=0)=\alpha$, is given by~\cite{cover}:
\beqa\label{asycap.eq}
  \lefteqn{I(\xi,\xih)\ =\ H_b(\alpha(1-\Pfa)+\balpha\Pmd)
  }
  \non
  &&
  \hspace*{2cm}
  - \alpha H_b(\Pfa) - \balpha H_b(\Pmd),
\eeqa
where $\balpha\triangleq1-\alpha$ and 
$H_b(p)\triangleq-p\log_2p-(1-p)\log_2(1-p)$ is the binary entropy 
function~\cite{cover}.

Finally, inserting \eqref{asycap.eq} into inequality \eqref{dpi.eq}, we obtain 
a relationship between the false-alarm and missed-detection probabilities, 
which represents a lower bound to the ROC for the given CR system.

The parametric dependence of the ROC lower bound on the mutual information is 
illustrated in Fig.\ \ref{genroc.fig}.
As expected, as $\I\uparrow1$, the ROC lower bounds decrease monotonically to 
$\Pfa=\Pmd=0$.

\insertfig{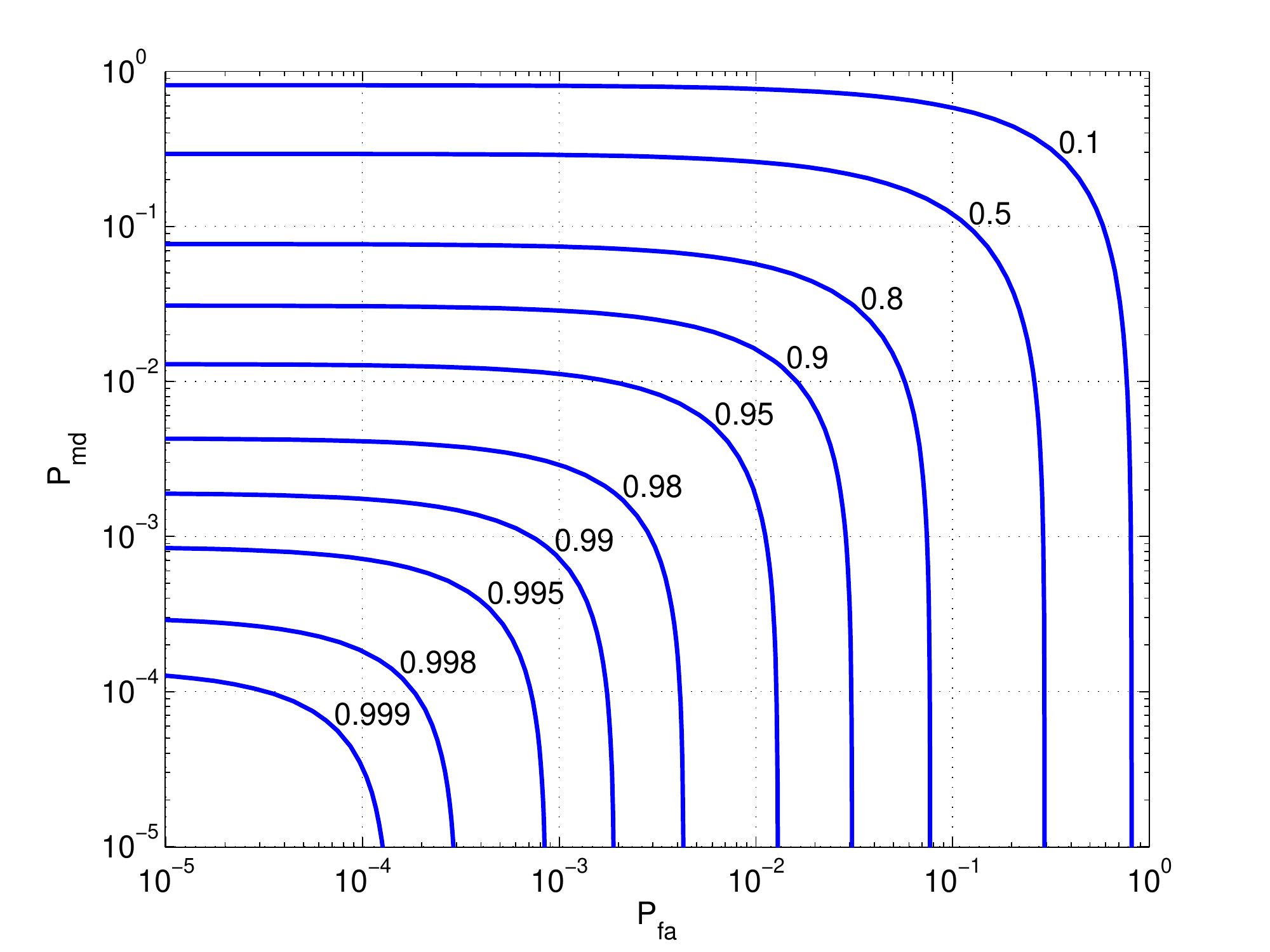}{
ROC lower bound curves corresponding to $\alpha=0.5$ and mutual information 
$\I$ indicated by the labels.}

\section{Calculation of $\I$}\label{mi.sec}

Let us define for convenience the following matrices and vectors:
\[
\left\{\begin{array}{lll}
  \Ym &\triangleq&  (y_k(n))_{k=1,n=1}^{K,N} \\[3mm]
  \Zm &\triangleq&  (z_k(n))_{k=1,n=1}^{K,N} \\[3mm]
  \hv &\triangleq&  (h_1,\dots,h_K)^\T \\
  \sv &\triangleq&  (s(1),\dots,s(N))^\H.
\end{array}\right..
\]
Then, we can simplify \eqref{xk.eq} by writing it as follows:
\beq\label{Y.eq}
  \Ym = \xi\ \hv\sv^\H+\Zm,
\eeq
and hence the mutual information \eqref{ycap.eq} becomes
\[
  \I = h(\Ym)-h(\Ym\mid\xi),
\]
where $h(\cdot)$ denotes the differential entropy~\cite{cover}.
First, it is plain to see that
\[
  h(\Ym\mid\xi)=h(\Zm)=N\sum_{k=1}^K\log_2(\pi\e\sigma_k^2).
\]
The evaluation of $h(\Ym)$ is more difficult.
We distinguish among different assumptions concerning the distribution of the 
secondary channel gain vector $\hv$ and the signal vector $\sv$.
In the following we consider the cases of $i)$ known gains and signal at the 
receiver, and of $ii)$ known gain and signal distribution at the receiver.

\subsection{Known gains and signal at the receiver}
\label{known.gains.signal.sec}

In order to equalize the noise variances, we transform the channel equation 
\eqref{Y.eq} by pre-multiplying by the inverse of the square root of the noise 
covariance matrix
\[
  \Sgm_z \triangleq \diag(\sigma_1^2,\dots,\sigma_K^2).
\]
We obtain
\beq\label{Y.eq1}
  \Ym = \xi\Am+\Zm,
\eeq
where $\Am\triangleq\Sgm_z^{-1/2}\hv\sv^\H$ and the entries of $\Zm$ are then 
iid as $\N_c(0,1)$.
This linear transformation is invertible and does not change the mutual 
information $\I$.
In order to calculate the mutual information, we resort to Theorem 
\ref{binchan.eq.th} (Appendix \ref{equiv.app}).
Since $\xi=0,1$, in order to use this result we can subtract $\Am/2$ and obtain 
symmetric input.
Theorem \ref{binchan.eq.th} tells us that the channel is equivalent to a 
binary-input real additive Gaussian channel with SNR $\|\Am\|^2/2$.
We obtain
\beqa\label{MI.eq}
  \I
  &=&
  H_b(\alpha)
  -\alpha\E\bigg[\log_2\bigg(1+\frac{\balpha}{\alpha}\e^{Z-\|\Am\|^2}\bigg)\bigg]
  \non
  &&
  -\balpha\E\bigg[\log_2\bigg(1+\frac{\alpha}{\balpha}\e^{Z-\|\Am\|^2}\bigg)\bigg],
\eeqa
where $Z\sim\N(0,2\|\Am\|^2)$%
\footnote{
Notation $Z\sim\N(\mu,\sigma^2)$ represents a real Gaussian random variable 
with mean $\mu$ and variance $\sigma^2$.
}

\begin{remark}
It is plain to see that \eqref{MI.eq} is invariant to the mapping 
$\alpha\mapsto1-\alpha$, i.e., to exchanging the {\em a priori} probabilities 
of primary signal presence and absence.
The ROC performance improves as these probabilities get closer to $0$ or to 
$1$, as illustrated in the following.
The symmetry of the resulting ROC lower bound suggests to define an {\em 
equilibrium point} corresponding to $\Pfa=\Pmd$, which is referred to as {\em 
equilibrium probability} and denoted by $\Peq$ in the sequel.
Under these operating conditions, the binary channel $\xi\to\xih$ is 
symmetric.
\end{remark}

\begin{remark}
It is worth noting that the mutual information $\I$, and hence the lower bound 
to the ROC, depend only on $\|\Am\|^2$ (in this case).
This parameter can be written as
\beq\label{snr.eq}
  \SNR \triangleq \|\Am\|^2 = \sum_{k=1}^K\frac{|h_k|^2\|\sv\|^2}{\sigma_k^2},
\eeq
which corresponds to the sum of the secondary users' receive SNR's.
For this reason, we refer to it in the following by the term {\em additive 
SNR}.
\end{remark}

\subsubsection{Limiting behavior for $\alpha\to0$}

Expanding \eqref{MI.eq} for $\alpha\to0$ we obtain:
\beqa
  \I
  &=&
  \E[1-Z+\SNR-\e^{Z-\SNR}]\alpha\log_2\e
  \non
  &&+
  \frac{\E[1-2\e^{\SNR-Z}+\e^{2Z-2\SNR}]}{2\log(2)}\alpha^2+
  O(\alpha^3)
  \non
  &=&
  \bigg[\SNR\ \alpha-
  \frac{1}{2}(\e^{2\SNR}-1)\alpha^2\bigg]\log_2\e+
  O(\alpha^3).
  \nonumber
\eeqa
We can see that the first- and second-order approximations represent upper and 
lower bounds, respectively, to the mutual information $\I$.
These lower bounds are illustrated in Fig.\ \ref{asy_ratio.fig}, plotting the 
ratio $\I/\alpha$ versus $\alpha$ and its second-order approximation (lower 
bound) for $\SNR=1$ (0~dB).
\insertfig{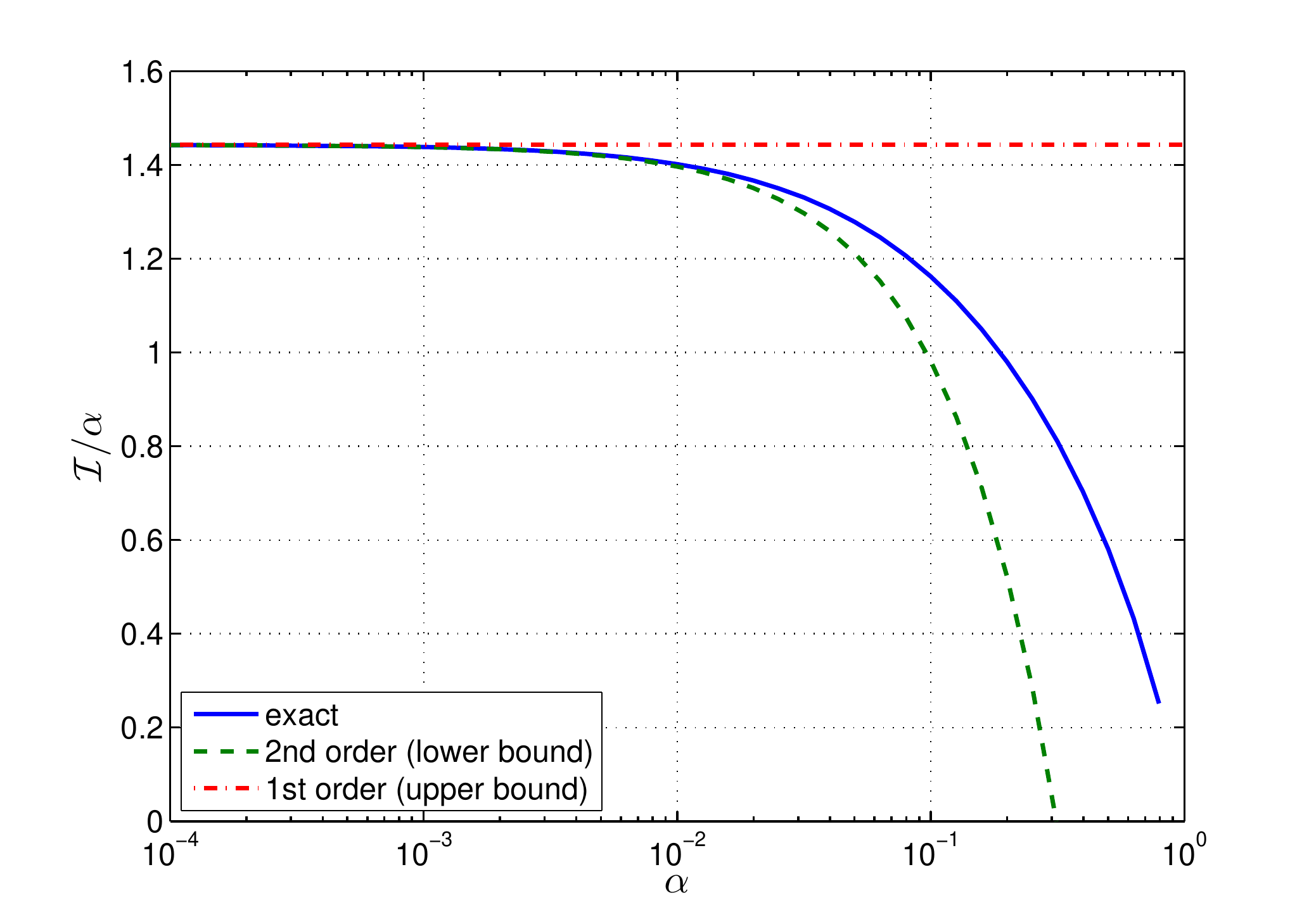}{
Plot of the ratio $\I/\alpha$ and of its second-order approximation
$[\SNR-0.5(\e^{2\SNR}-1)\alpha]\log_2\e$
versus $\alpha$ for $\SNR=1$ (0~dB).
}

\subsubsection{Limiting behavior for $\SNR\to\infty$}

Applying Theorem \ref{asy.bin.mi.th} (Appendix \ref{asy.app}), we obtain the 
following bounds:
\beqa
  \lefteqn{
  H_b(\alpha) - 2\frac{\sqrt{\pi\alpha\balpha}}{\ln2}
  \frac{1}{\SNR^{1/2}}\ \e^{-\SNR/4}
  \ \le\ \I
  }
  \non
  &\le&
  H_b(\alpha) - 2\frac{\sqrt{\pi\alpha\balpha}}{\ln2}
  \frac{1}{\SNR^{1/2}}\ \e^{-\SNR/4}
  \non
  &&
  +\frac{\sqrt{\pi\alpha\balpha}[\pi^2+8+(\ln(\alpha/\balpha))^2]}{2\ln2}
  \frac{1}{\SNR^{3/2}}\ \e^{-\SNR/4}.
  \nonumber
\eeqa
These bounds yield, for $\SNR\to\infty$, the following asymptotic 
approximation:
\beq\label{MI.asy.eq}
  \I \sim H_b(\alpha) - 2\frac{\sqrt{\pi\alpha\balpha}}{\ln2}
  \frac{1}{\SNR^{1/2}}\ \e^{-\SNR/4}.
\eeq

\subsubsection{Limiting behavior for $\alpha\to0$ and $\SNR\to\infty$}

Finally, we can also expand \eqref{asycap.eq} for $\alpha\to0$ at the 
equilibrium point of the ROC, and obtain
\[
  I(\xi;\xih) = (1-2\Peq)\log_2\frac{1-\Peq}{\Peq}~\alpha-
  \frac{(1-2\Peq)^2}{2\Peq(1-\Peq)}\alpha^2+O(\alpha^3).
\]
When both $\alpha$ and $\Peq\to0$, we obtain the approximation
\beq\label{Peq.asy.eq}
  \Peq\ \approx\ \e^{-\SNR}.
\eeq

\subsection{Known gain and signal distribution at the receiver}

The case of known gain and signal distribution at the receiver can be handled 
by exploiting the results derived in Section \ref{known.gains.signal.sec}.
First, we notice that $\Am=\Sgm_z^{-1/2}\hv\sv^\H$ is a random matrix whose 
joint pdf of the entries depends on the distributions of the channel gain 
vector $\hv$ and of the signal vector $\sv$.
Then, starting from \eqref{MI.eq}, we can apply the chain rule for the mutual 
information and the independence between $\xi$ and the vectors $\hv,\sv$ in 
order to obtain the following result:
\beqa
  \I
  &=&
  I(\xi;\Am)+I(\xi;\Ym\mid\Am)
  \non
  &=&
  I(\xi;\Ym\mid\Am)
  \non
  &=&
  H_b(\alpha) - \E\bigg[
  \alpha\log_2\bigg(1 + \frac{\balpha}{\alpha}
  \e^{\sqrt{2\Gamma}Z_1-\Gamma}\bigg)
  \non
  &&
  +\ \balpha\log_2\bigg(1 + \frac{\alpha}{\balpha}
  \e^{\sqrt{2\Gamma}Z_1-\Gamma}\bigg)\bigg].
\eeqa
where $Z_1\sim\N_c(0,1)$ is independent of $\Gamma\triangleq\|\Am\|^2$, and the 
average is with respect to both $Z_1$ and $\Gamma$.
In accordance with eq.\,\eqref{snr.eq}, we have
\[
  \Gamma = \sum_{k=1}^K\frac{|h_k|^2\|\sv\|^2}{\sigma_k^2},
\]
but in this framework $\Gamma$ is a random variable whose mean value is defined 
as the additive SNR, i.e.,
\[
  \SNR \triangleq \E[\Gamma].
\]
The lower bound to the ROC depends on the distribution of $\Gamma$.
Some examples illustrate this dependence in Section \ref{results.sec}.

\section{Numerical examples}\label{results.sec}

In this section we illustrate the ROC bound obtained by numerical examples in 
order to compare the lower bounds obtained with some real estimation scheme.

\subsection{Known gains and signal at the receiver}

The first example reported in Fig.\ \ref{roc1.fig}, which consider the case of 
known channel gains and signal (Section \ref{known.gains.signal.sec}) with {\em 
a priori} probability of primary signal absence $\alpha=0.5$.
It can be noticed that the curves are symmetric with respect to exchanging the 
probabilities $\Pfa$ and $\Pmd$.
We can also notice a threshold behavior with respect to the $\SNR$, which is 
better illustrated in Fig.\ \ref{roc2.fig}.
The SNR threshold lies between $5$ and $10$~dB:
below the threshold, the equilibrium probability decreases slowly;
above the threshold, the decrease rate becomes faster.
The curves in Fig.\ \ref{roc2.fig} are lower bounds to the equilibrium 
probability versus the additive SNR for different values of the probability of 
signal absence (or presence) and in the asymptotic case of $\alpha\to0$, which 
is given by eq.~\eqref{Peq.asy.eq}.

\insertfig{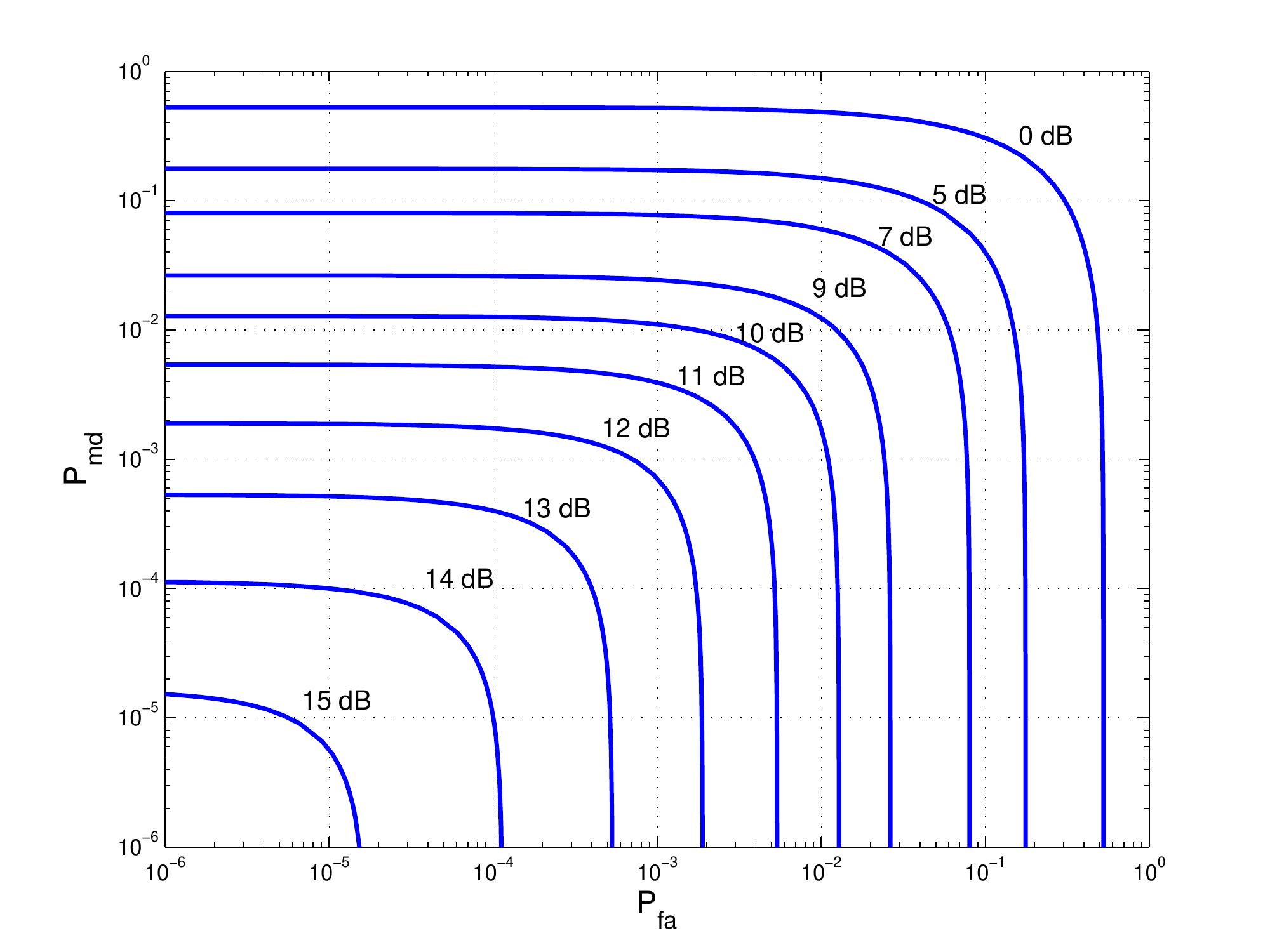}{ROC lower bound in the case of known channel gains and signal 
for different values of the additive SNR (reported on the plot) and {\em a 
priori} probability of primary signal absence $\alpha=0.5$.}

\insertfig{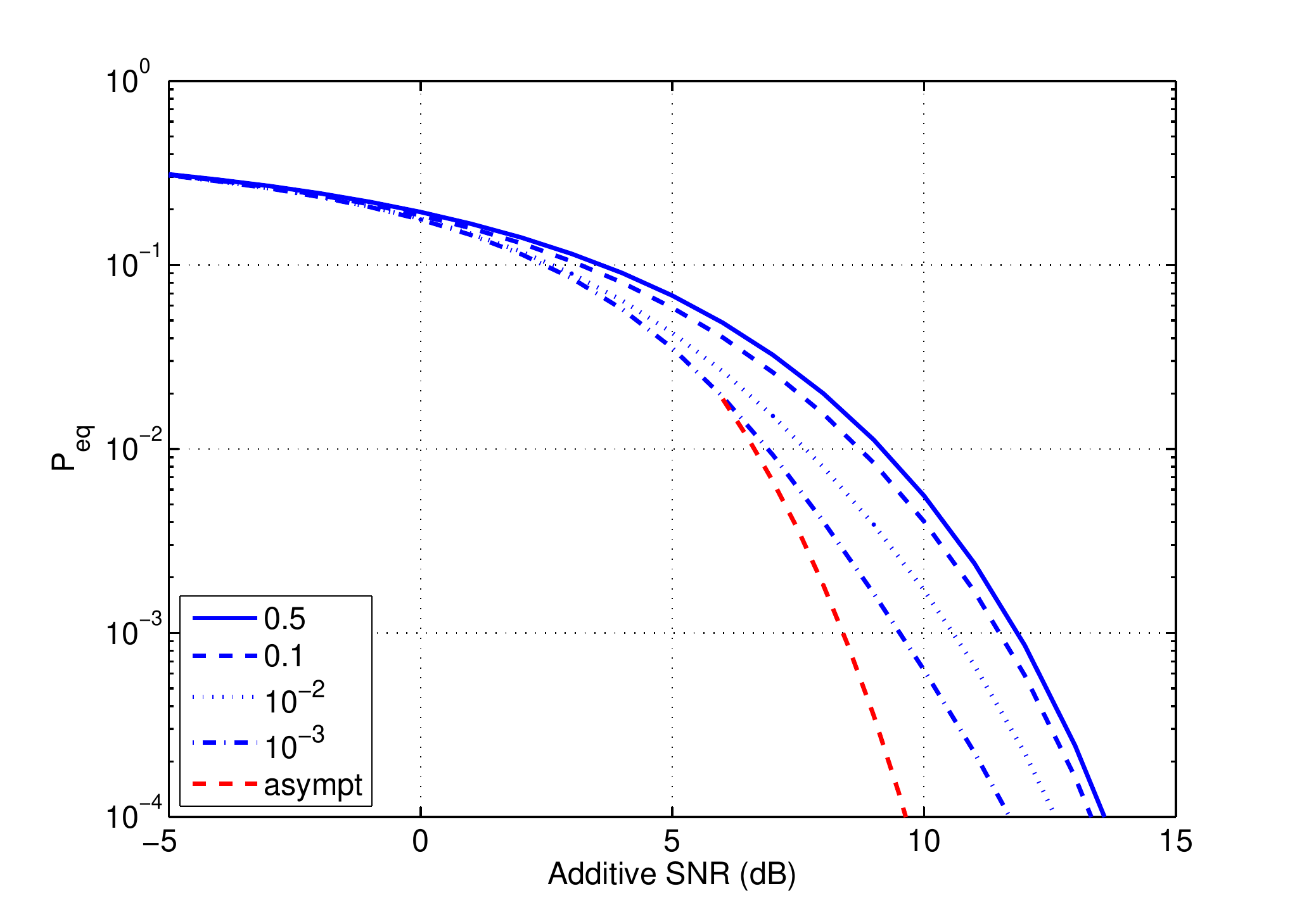}{Lower bound to the equilibrium probability $\Peq$ versus the 
additive SNR with {\em a priori} probability of primary signal absence (or 
presence) $\alpha=0.5,0.1,10^{-2},10^{-3}$ (solid curve).
The dashed curve corresponds to $\alpha\to0$ (asymptotic case).}

\subsection{Comparison with an energy detection scheme}

A simple spectrum sensing scheme based on energy detection corresponds to the 
following estimation rule:
\beq
  \xih = \left\{\begin{array}{ll}
    0, & \|\Ym\|^2<\theta+\ln\alpha \\
    1, & \|\Ym\|^2>\theta+\ln\balpha \\
  \end{array}\right.
\eeq
From the equivalent channel equation \eqref{Y.eq1}, the resulting false-alarm 
and missed-detection probabilities are given by:
\beq
  \xih = \left\{\begin{array}{lll}
    \Pfa &=& P(\|\Zm\|^2>\theta+\ln\balpha) \\
    \Pmd &=& P(\|\Am+\Zm\|^2<\theta+\ln\alpha)
  \end{array}\right.
\eeq
Since $\|\Zm\|^2$ and $\|\Am+\Zm\|^2$ are central and noncentral 
$\chi^2$-distributed random variables, we can find explicit expressions of the 
two probabilities.
In fact, the cdf's can be found in standard textbooks, such as~\cite{proakis}.
We have:
\[
  P(\|\Zm\|^2<u) = \gamma(KN,u),
\]
where $\gamma(n,x)\triangleq\Gamma(n)^{-1}\int_x^\infty u^{n-1}\e^{-u}du$ is 
the normalized upper incomplete Gamma function, and
\[
  P(\|\Am+\Zm\|^2<u) = 1-Q_{KN}(\sqrt{2\|\Am\|^2},\sqrt{2u}),
\]
where
$Q_m(a,b)\triangleq\int_b^\infty x(x/a)^{m-1}\e^{-(x^2+a^2)/2}I_{m-1}(ax)dx$
is the generalized Marcum's $Q$ function defined in~\cite{proakis}.

Figures \ref{roc3_0db.fig} and \ref{roc3_10db.fig} show the ROC corresponding 
to an energy detector spectrum sensing scheme for two values of the product 
$KN$ and $\SNR=0$ and $10$~dB, respectively.
The diagrams also report the information-theoretic lower bound derived in 
Section \ref{known.gains.signal.sec}.
We can see that increasing the product $KN$ for a fixed $\SNR$ degrades the 
resulting ROC.
This can be understood by observing that the variances of $\|\Am+\Zm\|^2$ and 
$\|\Zm\|^2$ are proportional to $KN$ (they are $(1+2\,\SNR)KN$ and $KN$, 
respectively), while the mean value difference is equal to $\SNR$.
Therefore, as the $KN$ increases, the overlapping of the two pdf's increases, 
and hence the probabilities of false-alarm and missed-detection.

\insertfig{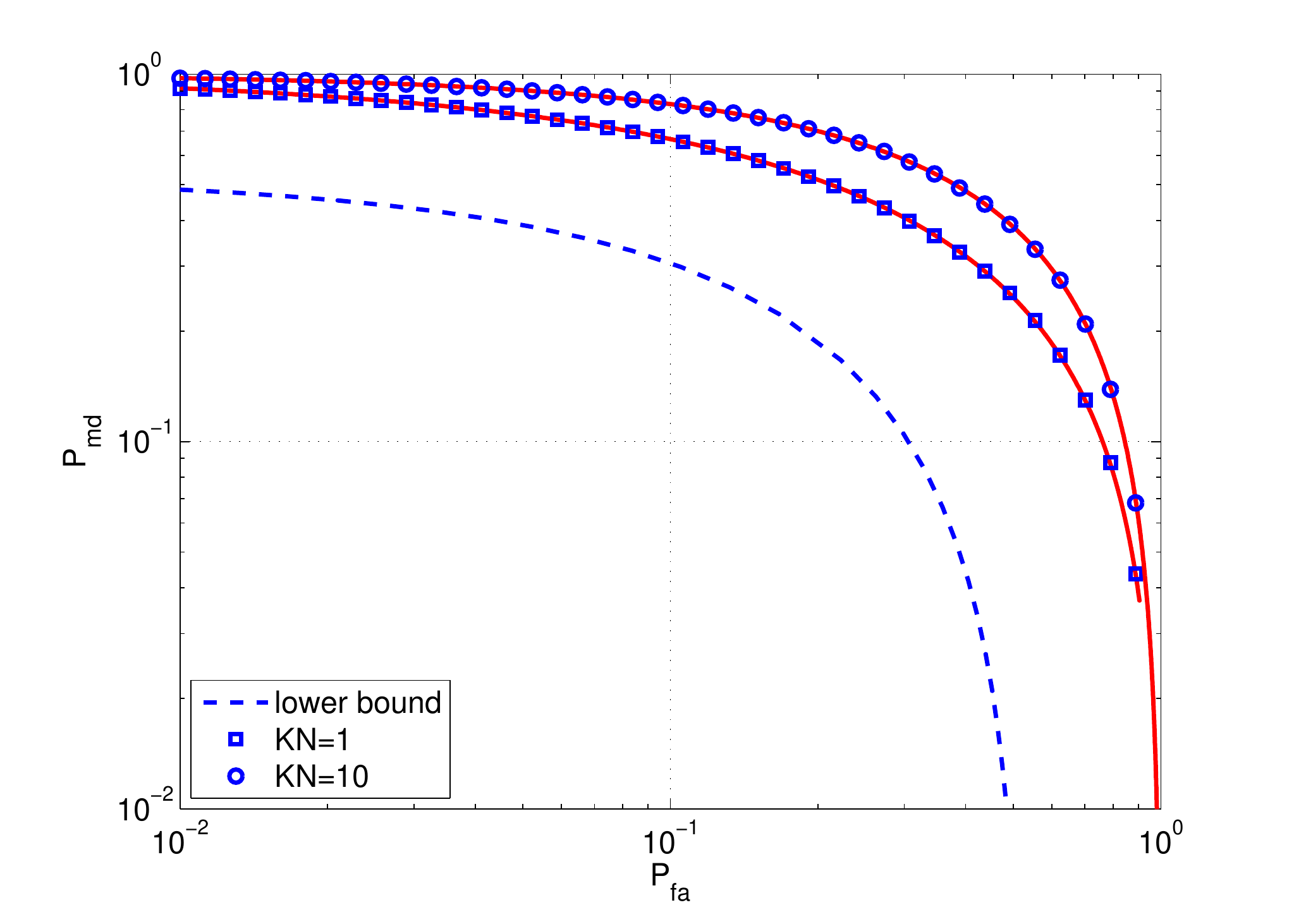}{ROC obtained with an energy detector with $\SNR=0$ dB, 
$KN=1$ and $10$, and $\alpha=0.5$.
Solid lines are obtained analytically and markers correspond to Monte-Carlo 
simulation results.
The lowest dashed curve corresponds to the information-theoretic lower bound.}

\insertfig{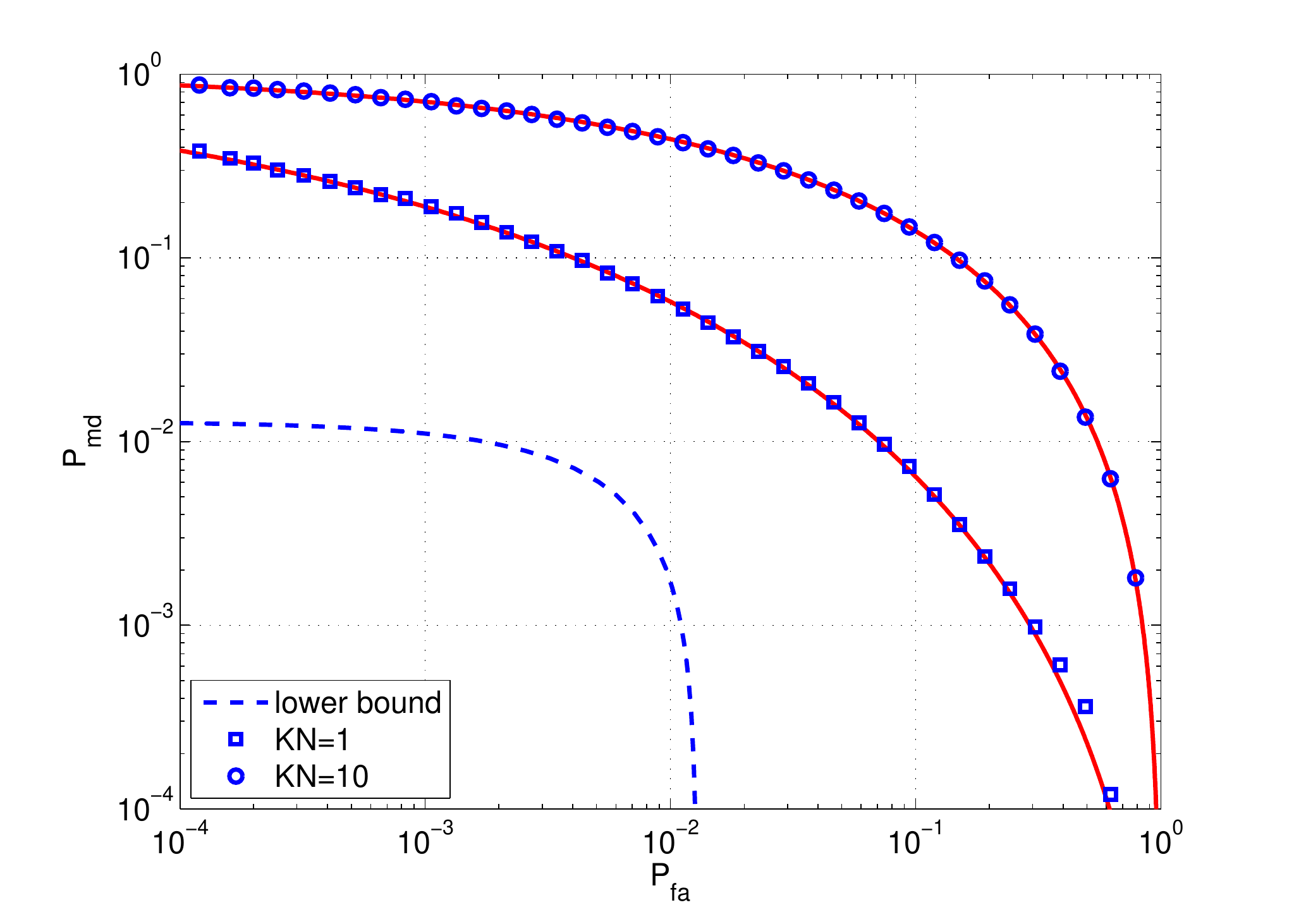}{Same as Fig.\ \ref{roc3_0db.fig} but $\SNR=10$ dB.}

\begin{remark}\label{snr.remark}
It is worth noting that the previous results hold for fixed additive SNR.
Then, increasing either $K$ or $N$ implies that the individual secondary user 
SNR's $|h_k|^2\|\sv\|^2/\sigma_k^2$ must decrease to keep the overall additive 
SNR constant.
On the contrary, if one fixes the individual SNR's, the additive SNR increases 
and both the lower bound and the energy-detector ROC improve.
Thus, the fact that the ROC curves decrease as $KN$ increases shall be 
interpreted by saying that the energy detector performance would improve if we 
could concentrate all the available sensors in a single one by keeping the 
total additive SNR constant.
\end{remark}

\subsection{Known gain and signal distribution at the receiver}

Here we consider the case of iid Rayleigh fading gains, where 
$\gamma_k\triangleq\E[|h_k|^2]/\sigma_k^2$, and $\|\sv\|^2$ has probability 
distribution $P(\|\sv\|^2=S_m)=p_m$ for $m=1,\dots,M$.
If we assume that all the $\gamma_k$ are different, the pdf of $\Gamma$ can be 
derived as follows:
\[
  p_\Gamma(G) = \sum_{m=1}^Mp_m\sum_{k=1}^K
  \frac{\exp(-G/(\gamma_kS_m))}{\gamma_kS_m}
  \prod_{\ell\ne k}\frac{1}{1-\gamma_\ell/\gamma_k}.
\]
We can use this result to calculate the double integral
\beqa\label{MI.ray.eq}
  \I
  &=&
  H_b(\alpha) -
  \int_{-\infty}^\infty \frac{\exp(-z^2/2)}{\sqrt{2\pi}}
  \int_0^\infty p_\Gamma(G)
  \non
  &&
  \bigg[
  \alpha\log_2\bigg(1 + \frac{\balpha}{\alpha}\e^{\sqrt{2G}z-G}\bigg)
  \non
  &&
  +\ \balpha\log_2\bigg(1 + \frac{\alpha}{\balpha}\e^{\sqrt{2G}z-G}\bigg)
  \bigg]dG\,dz.
\eeqa
As an illustrative example, we consider the following scenario:
\bit
\item $K=4$ secondary users.
\item $\gamma_k=(4+k)$ dB for $k=1,\dots,K$.
\item $\|\sv\|^2=1$.
\eit
The ROC curves are reported in Fig.\ \ref{roc_ray.fig}.
The lowest curve corresponds to the lower bound calculated by using 
\eqref{MI.ray.eq}.
The other curves correspond to the implementation of a spectrum sensing 
algorithm based on energy detection for different combinations of the number of 
secondary users $K$ and sampling times $N$.
In all cases, the same additive SNR is assumed, $\SNR=\sum_{k=1}^K\gamma_k$, 
namely,
\[
  10\log_{10}(10^{0.5}+10^{0.6}+10^{0.7}+10^{0.8})=12.66\ \mathrm{dB}.
\]
As already noticed in Remark \ref{snr.remark}, the best operating condition for 
the energy detector corresponds to the case of $K=N=1$ (at fixed additive 
SNR).

\insertfig{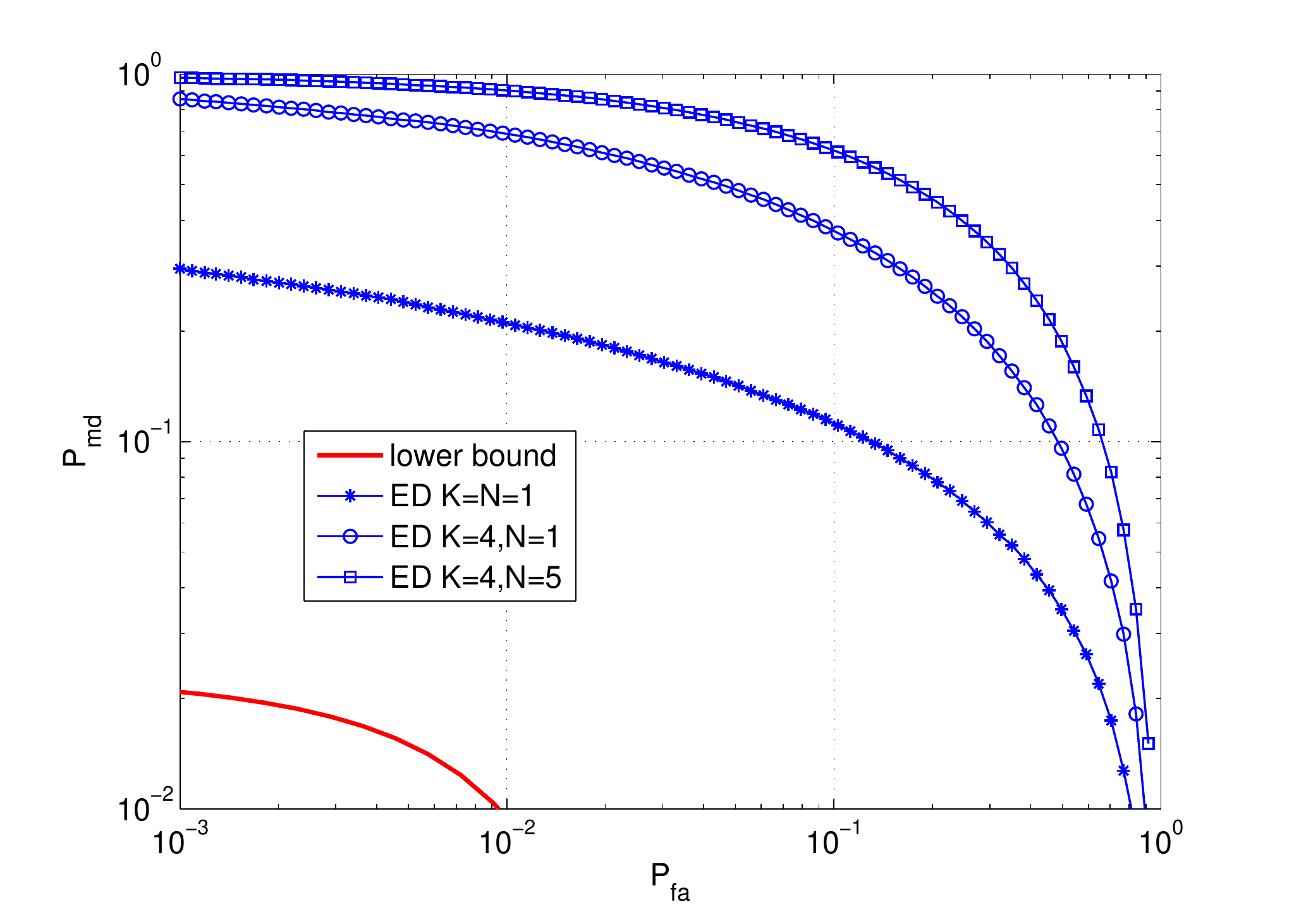}{
ROC curves corresponding to $\alpha=0.5$, $K=4$ secondary users, and Rayleigh 
fading.
The solid curve is the lower bound.
The other curves correspond to energy detection (ED) with different 
combinations of $K$ and $N$ and constant additive SNR ($12.66$ dB).
}

\section{Conclusions}\label{concl.sec}

In this work we proposed an information-theoretic method to derive a lower 
bound to the receiver operating characteristic of a cognitive radio network 
based on cooperative sensors.
The bound stems from the application of the data-processing inequality to the 
binary asymmetric channel arising by considering the primary signal presence as 
a binary input and the fusion center decision on the signal presence as a 
binary output.
The bound takes into account the possible knowledge of the {\em a priori} 
probability of primary signal presence and applies to every kind of 
single-input multiple-output channel connecting the primary signal to the 
multiple cooperative sensors (i.e., the secondary users of the cognitive radio 
system).

Key advantages of this approach are: $i)$ independence from the implementation 
of the connection between the sensors and the fusion center; and $ii)$ 
independence from the fusion rule.
Both features derive from the information-theoretic method we have followed, 
based on the equivalence between the actual channel model (connecting the 
primary transmitter to the sensors and then to the fusion center) and the 
binary asymmetric channel with error probabilities corresponding to the 
false-alarm and missed-detection events.

In order to illustrate this basic idea, we considered two scenarios of interest 
for cognitive radio networks:
\ben
\item
The case of full channel gain and primary signal information at the fusion 
center.
\item
The case of full distribution information about the channel gain and the 
primary signal at the fusion center.
\een

The first case has been investigated in full detail by deriving the mutual 
information $\I$ between the primary signal presence variable $\xi$ and the set 
of sensor observations $\Ym$.
This expression has been analyzed asymptotically, both for large additive SNR 
and for vanishing $\alpha$ (probability of signal absence).
The asymptotic expression for large additive SNR, eq.~\eqref{MI.asy.eq}, is 
based on an integral which is an extension of the one calculated 
in~\cite[Prob.~4.12]{rich-urb}.
In essence, this expression can lead to an asymptotic expansion of the mutual 
information of the binary symmetric channel for large SNR.
Full details on the derivation are reported in Appendix~\ref{asy.app}.

In the second case, a general expression of the mutual information $\I$ 
required to obtained the receiver operating characteristic lower bound has been 
derived as the average value of an expression depending on two random variables 
(the random additive SNR $\Gamma$ and the auxiliary Gaussian random variable 
$Z_1$).
This average leads to a double integral, which has been expanded in detail in 
the case of iid Rayleigh distributed channel gains.
In the numerical results section, the distribution of $\Gamma$ is reported in a 
fairly general case, and numerical results are included for illustration 
purposes.

Finally, it is worth mentioning that the series expansion of the mutual 
information of a binary input additive Gaussian channel 
(Theorem~\ref{asy.bin.mi.th}) is also a novel contribution of this paper.
It extends the series expansion of the capacity of the same channel, which can 
be found in~\cite[Prob.~4.12]{rich-urb}.
In the present context, this series expansion is needed to account for the 
possible knowledge of the {\em a priori} probability of primary signal 
presence, which is available in several cognitive radio system and worth being 
used to improve the quality of the decision rule at the fusion center.

\appendices
\section{Equivalence of binary input Gaussian channels}\label{equiv.app}

In this appendix we show that a binary-input, vector-output additive Gaussian 
channel is equivalent, as far as concerns the mutual information, to a 
binary-input additive Gaussian channel with scalar output.

Let the binary input be $X=\pm1$ with $\alpha\triangleq P(X=-1)$ and 
$\balpha\triangleq1-\alpha=P(X=+1)$.

Let the channel equation be
\beq
  \yv = X\av+\zv,
\eeq
where $\av\in\R^{n\times1}$ is a given constant vector and $\zv$ is a vector of 
iid Gaussian random variables distributed as $\N(0,1)$.

The mutual information is given by
\[
  I(X;\yv) = h(\yv)-h(\yv|X) = h(\yv)-h(\zv).
\]
We know that $h(\zv)=(n/2)\log_2(2\pi\e)$.
To calculate $h(\yv)$, we note that
\[
  p_\yv(\yv) =
  \frac{\alpha\exp(-\|\yv+\av\|^2/2)+\balpha\exp(-\|\yv-\av\|^2/2)}{(2\pi)^{n/2}}.
\]
Hence,
\beqa\label{bin.vec.mi.eq}
  \lefteqn{I(X;\yv)}
  \non
  &=&
  -\E\bigg[\log_2(\alpha\e^{-\|\yv+\av\|^2/2}+\balpha\e^{-\|\yv-\av\|^2/2})\bigg]
  -\frac{n\log_2\e}{2}
  \non
  &=&
  H_b(\alpha)
  -\alpha\,\E\bigg[\log_2\bigg(1+\frac{\balpha}{\alpha}\e^{2\av^\T\zv-2\|\av\|^2}\bigg)\bigg]
  \non
  &&
  -\balpha\,\E\bigg[\log_2\bigg(1+\frac{\alpha}{\balpha}\e^{2\av^\T\zv-2\|\av\|^2}\bigg)\bigg].
\eeqa
The scalar product $\av^\T\zv$ is a real Gaussian random variable with zero 
mean and variance $\|\av\|^2$.

Similarly, we can find the mutual information of the channel
\beq
  Y = aX+Z,
\eeq
where $a\in\R$ and $Z\in\N(0,1)$.
In this case,
\beqa\label{bin.scal.mi.eq0}
  \lefteqn{I(X;Y)}
  \non
  &=&
  -\E\bigg[\log_2(\alpha\e^{-(Y+a)^2/2}+\balpha\e^{-(Y-a)^2/2})\bigg]
  -\frac{\log_2\e}{2}
  \non
  &=&
  H_b(\alpha)
  -\alpha\,\E\bigg[\log_2\bigg(1+\frac{\balpha}{\alpha}\e^{2aZ-2a^2}\bigg)\bigg]
  \non
  &&
  -\balpha\,\E\bigg[\log_2\bigg(1+\frac{\alpha}{\balpha}\e^{2aZ-2a^2}\bigg)\bigg].
\eeqa
Then, the mutual information eqs.\ \eqref{bin.vec.mi.eq} and 
\eqref{bin.scal.mi.eq0} coincide provided that $a^2=\|\av\|^2$.
These results are summarized by the following theorem.

\begin{theorem}\label{binchan.eq.th}
A binary-input vector-output additive Gaussian channel $\yv=X\av+\zv$, where 
the entries of $\zv$ are iid and distributed as $\N(0,1)$, is equivalent, in 
terms of mutual information, to a binary-input real additive Gaussian channel 
with SNR $\gamma=\|\av\|^2$.
The mutual information is:
\beqa\label{bin.scal.mi.eq}
  I(X;Y)
  &=&
  H_b(\alpha)
  -\alpha\,\E\bigg[\log_2\bigg(1+\frac{\balpha}{\alpha}
  \e^{2\sqrt\gamma Z_1-2\gamma}\bigg)\bigg]
  \non
  &&
  -\balpha\,\E\bigg[\log_2\bigg(1+\frac{\alpha}{\balpha}
  \e^{{2\sqrt\gamma Z_1-2\gamma}}\bigg)\bigg],
\eeqa
where $Z_1\sim\N(0,1)$.
\end{theorem}

\begin{remark}
The equivalence between the scalar and $n$-vector channels in terms of mutual 
information can be predicted by observing that the vector channel is a 
combination of $n$ parallel Gaussian channels.
Then, provided the receiver knows the vector $\av$, the useful part of the 
signal can be combined coherently at the receiver while the noise is combined 
incoherently.
This implies an $n$-fold increase of the SNR, for a given mutual information, 
when passing from the scalar to the $n$-vector channel.
\end{remark}

\section{Asymptotic approximation of \eqref{bin.scal.mi.eq}}\label{asy.app}

In this appendix we derive an asymptotic approximation of the mutual 
information \eqref{bin.scal.mi.eq}.
This result extends \cite[Prob.~4.12]{rich-urb}, corresponding to the 
equiprobable input case.

\begin{theorem}\label{asy.bin.mi.th}
The asymptotic expansion of \eqref{bin.scal.mi.eq} for $\gamma\to\infty$ is 
given by
\beq\label{asy.bin.mi.eq}
  I(X;Y) = H_b(\alpha) - \sum_{n=0}^\infty
  \frac{(-1)^nk_n(\alpha)}{\gamma^{n+1/2}}\ \e^{-\gamma/2},
\eeq
where
\beqa\label{asy.bin.mi.coef}
  k_n(\alpha)
  &=&
  \frac{2\sqrt{\pi\alpha\balpha}}{n!\ln2}
  \sum_{k=0}^n\binom{2n}{2k} \pi^{2(n-k)} |E_{2(n-k)}|
  \non
  &&\times
  \sum_{m=0}^{k}2^{2m}(2k)^{(2m)}(\ln(\alpha/\balpha))^{2(k-m)}.
\eeqa
Consecutive partial sums in \eqref{asy.bin.mi.eq} are lower and upper bounds to 
$I(X;Y)$.
\end{theorem}

\begin{proof}
We start by considering the integral
\beqa\label{series.eq}
  \lefteqn{
  \int_{-\infty}^\infty\ln(1+\rho\ \e^{z-2\gamma})
  \e^{-z^2/(8\gamma)}\frac{dz}{\sqrt{8\pi\gamma}}
  }
  \non
  &=&
  \int_{-\infty}^\infty\ln(1+\rho\ \e^z)
  \e^{-(z+2\gamma)^2/(8\gamma)}\frac{dz}{\sqrt{8\pi\gamma}}
  \non
  &=&
  \int_{-\infty}^\infty\ln(1+\rho\ \e^z)
  \e^{-z^2/(8\gamma)-z/2-\gamma/2}\frac{dz}{\sqrt{8\pi\gamma}}
  \non
  &=&
  2\sqrt\pi\sum_{n=0}^\infty \frac{(-1)^n\ c_n(\rho)}{(8\gamma)^{n+1/2}}\ 
\e^{-\gamma/2}
\eeqa
where
\beq
  c_n(\rho) \triangleq \frac{1}{2\pi}
  \int_{-\infty}^\infty\ln(1+\rho\ \e^z)\ \frac{z^{2n}}{n!}\ \e^{-z/2} dz.
\eeq
Since consecutive partial sums of the series expansion of $\e^{-x}$ are lower 
and upper bound of the limit, also consecutive partial sums of 
\eqref{series.eq} are lower and upper bounds of the lhs.

To calculate the coefficients $c_n(\rho)$, we notice that
\beqa\label{cn1.eq}
  \lefteqn{c_n'(\rho)}
  \non
  &=&
  \frac{1}{4\pi\ n!\sqrt\rho}
  \int_{-\infty}^\infty\frac{z^{2n}}{\cosh((z+\ln\rho)/2)}dz
  \non
  &=&
  \frac{1}{2\pi\ n!\sqrt\rho}
  \int_{-\infty}^\infty\frac{(2u-\ln\rho)^{2n}}{\cosh u}du
  \non
  &=&
  \frac{1}{2\pi\ n!\sqrt\rho}
  \sum_{k=0}^n\binom{2n}{2k}2^{2k}(\ln\rho)^{2(n-k)}
  \int_{-\infty}^\infty\frac{u^{2k}}{\cosh u}du
  \non
  &=&
  \frac{1}{2\,n!}
  \sum_{k=0}^n\binom{2n}{2k}\pi^{2k}
  \frac{(\ln\rho)^{2(n-k)}}{\sqrt{\rho}}|E_{2k}|
\eeqa
by using the integral \cite[3.523-4]{grad}
\[
  \int_{-\infty}^\infty\frac{u^{2k}}{\cosh u}du = 
2\bigg(\frac{\pi}{2}\bigg)^{2k+1}|E_{2k}|,
\]
for every integer $k\ge0$,
where $E_n$ is the $n$th Euler number ($E_0=1,E_2=-1,E_4=5,E_6=-61,\dots$).
Since we have:
\[
  \int_0^x\frac{(\ln u)^n}{\sqrt{u}}du
  = 2\sqrt{x}\sum_{k=0}^n(-1)^k2^kn^{(k)}(\ln x)^{n-k},
\]
for every integer $n\ge0$ and with $n^{(m)}\triangleq(n!/(n-m)!)$, we can 
integrate \eqref{cn1.eq} term by term and obtain the following final 
expression:
\beqa\label{cn.rho.eq}
  c_n(\rho)
  &=&
  \frac{1}{2\,n!}
  \sum_{k=0}^n\binom{2n}{2k}\pi^{2k}|E_{2k}|
  \int_0^\rho\frac{(\ln u)^{2(n-k)}}{\sqrt{u}}du
  \non
  &=&
  \frac{\sqrt\rho}{n!}
  \sum_{k=0}^n\binom{2n}{2k} \pi^{2(n-k)} |E_{2(n-k)}|
  \non
  &&\times
  \sum_{m=0}^{2k}(-1)^m2^m(2k)^{(m)}(\ln\rho)^{2k-m}.
\eeqa
In the special case of $\rho=1$, we have:
\[
  c_n(1) = \frac{1}{n!}
  \sum_{k=0}^n2^{2k}(2n)^{(2k)}\pi^{2(n-k)} |E_{2(n-k)}|.
\]
Then, in accordance with \cite[Prob.~4.12]{rich-urb}, where $c_n$ corresponds 
to $c_n(1)/2^{2n}$, we have:
\beqa
  c_0(1) &=& 1 \non
  c_1(1) &=& 8+\pi^2 \non
  c_2(1) &=& 384+48\pi^2+5\pi^4 \non
  c_3(1) &=& 46080+5760\pi^2+600\pi^4+61\pi^6 \non
  &\vdots\nonumber
\eeqa

Finally, we apply \eqref{cn.rho.eq} to derive the asymptotic expansion 
\eqref{asy.bin.mi.eq} and the coefficients \eqref{asy.bin.mi.coef}.

\end{proof}


\end{document}